\documentclass[9pt]{IEEEtran}
\usepackage{cite,physics,esdiff,mathtools}
\usepackage{amsmath,amssymb,amsfonts,amsthm,mathrsfs,comment,multirow}
\usepackage[inline]{enumitem}
\usepackage[ruled,linesnumbered]{algorithm2e}
\SetKw{Break}{break}
\newtheorem{theorem}{Theorem}
\newtheorem{lemma}{Lemma}
\newtheorem{remark}{Remark}

\newtheorem{corollary}{Corollary}
\newtheorem{problem}{Problem}

\newtheorem{proc}{Procedure}
\usepackage{graphicx}

\usepackage{textcomp}
\def\BibTeX{{\rm B\kern-.05em{\sc i\kern-.025em b}\kern-.08em
    T\kern-.1667em\lower.7ex\hbox{E}\kern-.125emX}}

\DeclareMathOperator{\vect}{vec}
\DeclareMathOperator{\svec}{svec}
\begin{document}
\title{Robust Policy Iteration for Continuous-time Linear Quadratic Regulation}
\author{Bo Pang, \IEEEmembership{Student Member, IEEE}, Tao Bian, \IEEEmembership{Member, IEEE} and Zhong-Ping Jiang, \IEEEmembership{Fellow, IEEE}
\thanks{This work has been supported in part by the U.S. National Science Foundation under Grants ECCS-1501044 and EPCN-1903781.}
\thanks{B. Pang and Z. P. Jiang are with the Control and Networks Lab, Department of Electrical and Computer Engineering, Tandon School of Engineering, New York University, 370 Jay Street, Brooklyn, NY 11201, USA (e-mail: bo.pang@nyu.edu; zjiang@nyu.edu).
}%
\thanks{T. Bian is with the Bank of America Merrill Lynch, One Bryant Park, New York, NY 10036, USA (e-mail: tbian@nyu.edu).}
}

\renewcommand{\arraystretch}{1.2}

\maketitle

\begin{abstract}
This paper studies the robustness of policy iteration in the context of continuous-time infinite-horizon linear quadratic regulation (LQR) problem. It is shown that Kleinman's policy iteration algorithm is inherently robust to small disturbances and enjoys local input-to-state stability in the sense of Sontag. More precisely,  whenever the disturbance-induced input term in each iteration is bounded and small,
the solutions of the policy iteration algorithm are also bounded and enter a small neighborhood of the optimal solution of the LQR problem. Based on this result, an off-policy data-driven policy iteration algorithm for the LQR problem is shown to be robust when the system dynamics are subjected to small additive unknown bounded disturbances. The theoretical results are validated by a numerical example.
\end{abstract}

\section{INTRODUCTION}
As an important iterative method in reinforcement learning (RL) and approximate/adaptive dynamic programming (ADP), policy iteration has been widely studied by researchers and utilized in different kinds of real-life applications by practitioners \cite{sutton2018reinforcement,jiang2017robust,Frank-book1,BertsekasOptimalControl2}. Policy iteration involves two steps, \textit{policy evaluation} and \textit{policy improvement}. In policy evaluation, a given control policy is evaluated based on a scalar performance index. Then this performance index is utilized to generate a new control policy in policy improvement. These two steps are iterated in turn, to find the solution of the optimal control problem at hand. When all the information involved in this process is exactly known, the convergence to the optimal solution can be provably guaranteed, by exploiting the monotonicity property of the policy improvement step. That is, the performance of the newly generated policy is no worse than that of the given policy in each iteration. Over the past decades, various versions of policy iteration have been proposed, for diverse optimal control problems (see \cite{sutton2018reinforcement,jiang2017robust,Frank-book1,BertsekasOptimalControl2} and the reference therein).

In reality, policy evaluation or policy improvement can hardly be implemented precisely, because of the existence of various errors that may be induced by function approximation, state estimation, measurement noise, external disturbance and so on. Therefore, a natural question to ask is: when is a policy iteration algorithm robust in the presence of disturbance inputs? In spite of the popularity of policy iteration, its robustness issue has not been fully understood yet, due to the inherent nonlinearity of the process and the (possible) loss of monotonicity in the presence of errors. The problem becomes more complex when the state and control input spaces are unbounded and continuous, and the stability issue needs to be addressed. Some results about the robustness of policy iteration for optimal control problem with unbounded state/control spaces are available in previous literature \cite{BertsekasOptimalControl2,boussios1998approach,hylla2011extension,pmlr-v89-abbasi-yadkori19a,krauth2019finite}. An error bound is derived in \cite[Proposition 3.6]{BertsekasOptimalControl2}, for the discounted optimal control of discrete-time systems. Sufficient conditions on the error are given in \cite[Chapter 2]{hylla2011extension} and \cite[Chapter 2]{boussios1998approach} for continuous-time linear and nonlinear systems respectively, to guarantee that the newly generated control policy is stabilizing and the monotonicity is preserved. Uniform boundedness of the solutions given by policy iteration in the presence of errors is established in \cite[Lemma 5.1]{pmlr-v89-abbasi-yadkori19a}, for the discrete-time LQR problem. Utilizing the fact that the discrete-time Riccati operator is contractive in certain metric, it is shown in \cite[Appendix B]{krauth2019finite} that the convergence to the optimal solution is still achieved, if the error converges to zero.

In this paper, we investigate the robustness of the first policy iteration algorithm proposed by Kleinman \cite{1098829} for the continuous-time LQR problem. By regarding the process of policy iteration as a discrete-time nonlinear system, we firstly prove that the optimal solution of the LQR problem is a locally exponentially stable equilibrium of the error-free policy iteration. Secondly, with the aid of this observation, we show that the policy iteration with the errors as the input is locally input-to-state stable \cite{sontag2008input}. That is, if the policy iteration starts from an initial solution close to the optimal solution, and the errors are small and bounded, the discrepancies between the solutions generated by our proposed policy iteration algorithm and the optimal solution will also be small and bounded. Thirdly, we demonstrate that for any initial stabilizing control gain, as long as the errors are small, the approximate solution given by policy iteration will eventually enter a small neighbourhood of the optimal solution. Finally, as an application of the derived results, an off-policy data-driven policy iteration algorithm originally proposed in \cite{JIANG20122699} is proven to be robust, when the system dynamics are disturbed by small additive unknown disturbances. Simulations on a numerical example validate our results.

The rest of this paper is organized as follows: Section \ref{section_problem_formulation} introduces the problem formulation and preliminaries. Section \ref{section_robust_analysis} presents the robustness analysis of Kleinman's model-based policy iteration algorithm. Section \ref{section_datadriven_PI} demonstrates the robustness of the data-driven off-policy policy iteration in \cite{JIANG20122699}. The efficacy of the robust policy iteration method is validated by a numerical example in Section \ref{section_simulation}. Section \ref{section_conclusion} closes the paper with some concluding remarks. The proofs of all the theoretical results in this paper are given in the Appendix.

\textbf{Notations.} $\mathbb{R}$ ($\mathbb{R}_+$) is the set of all real (nonnegative) numbers; $\mathbb{Z}_+$ denotes the set of nonnegative integers; $\mathbb{S}^{n}$ is the set of all real symmetric matrices of order $n$; $\otimes$ denotes the Kronecker product; $I_n$ denotes the identity matrix with dimension $n$; $\Vert \cdot\Vert_F$ is the Frobenius norm; $\Vert\cdot\Vert_2$ is the 2-norm for vectors and the induced 2-norm for matrices; for signal $u:\mathbb{F}\rightarrow \mathbb{R}^{n\times m}$, $\Vert u\Vert_\infty$ denotes its $l^\infty$-norm when $\mathbb{F} = \mathbb{Z}_+$, and $L^\infty$-norm when $\mathbb{F} = \mathbb{R}_+$. For matrices $X\in \mathbb{R}^{m\times n}$, $Y\in \mathbb{S}^m$, and vector $v\in \mathbb{R}^{n}$, define 
\begin{align*}
    \vect(X) &= [\begin{array}{cccc}
    X_1^T & X_2^T & \cdots & X_n^T
    \end{array}]^T, \\
    \svec(Y) &= [y_{11},\sqrt{2}y_{12},\cdots,\sqrt{2}y_{1m},y_{22},\sqrt{2}y_{23}, \\ 
    &\cdots,\sqrt{2}y_{m-1,m},y_{m,m}]^T\in \mathbb{R}^{\frac{1}{2}m(m+1)}, \\
    \tilde{v} &= \svec(vv^T),
\end{align*}
where $X_i$ is the $i$th column of $X$. For $Z\in\mathbb{R}^{m\times n}$, $U\in\mathbb{S}^{m+n}$, define
$$\mathcal{H}(U,Z) = \left[\begin{array}{cc}
                I_n &  -Z^T
                \end{array}\right] U\left[\begin{array}{c}
                I_n  \\
                -Z
                \end{array}\right].$$
Define $\mathcal{B}_{r}(Z) = \{X\in\mathbb{R}^{m\times n}\vert \Vert X - Z\Vert_F<r\}$ and $\mathcal{\bar{B}}_{r}(Z)$ as the closure of $\mathcal{B}_{r}(Z)$. $Z^\dagger$ is the Moore-Penrose inverse of matrix $Z$. A function $\gamma:\mathbb{R}_+\rightarrow\mathbb{R}_+$ is said to be of class $\mathcal{K}$ if it is continuous, strictly increasing and vanishes at the origin. A function $\beta:\mathbb{R}_+\times \mathbb{R}_+\rightarrow\mathbb{R}_+$ is said to be of class $\mathcal{K}\mathcal{L}$ if $\beta(\cdot,t)$ is of class $\mathcal{K}$ for every fixed $t\geq 0$ and, for every fixed $r\geq 0$, $\beta(r,t)$ decreases to $0$ as $t\rightarrow\infty$.

\section{Problem Formulation and Prelinimaries}\label{section_problem_formulation}
Consider linear time-invariant systems of the form
\begin{equation}\label{LTI_sys}
    \dot{x} = Ax + Bu,\quad x(0) = x_0
\end{equation}
where $x\in\mathbb{R}^n$ is the system state, $u\in\mathbb{R}^m$ is the control input, $x(0)=x_0\in\mathbb{R}^n$ is the initial condition, $A\in\mathbb{R}^{n\times n}$ and $B\in\mathbb{R}^{n\times m}$, $\left(A,B\right)$ is controllable. The classic LQR problem is to find a controller $u$ in order to minimize the following cost
\begin{equation}\label{cost_deterinistic}
    J(x_0,u) = \int_{0}^\infty x^TQx + u^TRu \dd t,
\end{equation}
where $Q\in\mathbb{S}^{n}$ is positive semidefinite, $(A,Q^{1/2})$ is observable and $R\in\mathbb{S}^{m}$ is positive definite. The LQR problem admits the optimal controller $u^* = -K^*x$, where $K^*=R^{-1}B^TP^*$ and $P^*\in\mathbb{S}^n$ is the unique positive definite solution of the algebraic Riccati equation (ARE)
\begin{equation}\label{ARE}
A^TP+PA-PBR^{-1}B^TP+Q=0,
\end{equation}
and $A-BK^*$ is Hurwitz. See \cite[Section 6.2.2]{10.2307/j.ctvcm4g0s} for details.

Kleinman's policy iteration algorithm aimed at solving (\ref{ARE}) iteratively is presented below, which is an equivalent reformulation of the original results in \cite{1098829}. For convenience, we say that a control gain $K\in\mathbb{R}^{n\times m}$ is stabilizing if $A-BK$ is Hurwitz.
\begin{proc}[Kleinman's Policy Iteration]\label{procedure_policy_itration}
\ \par
\begin{enumerate}[label=\arabic*)]
    \item Choose a stabilizing control gain $K_1$, and let $i=1$.
    \item\label{standard_PI_PE} (Policy evaluation) Evaluate the performance of control gain $K_i$, by solving 
    \begin{equation}\label{RPI_PE}
       \mathcal{H}(G_i,K_i)=0
    \end{equation}
    for $P_i\in\mathbb{S}^n$, where
    \begin{equation*}
    \begin{split}
        G_i &= \left[
                    \begin{array}{c|c}
                    G_{i,11} & G_{i,21}^T \\
                    \hline
                    G_{i,21} & G_{i,22}
            \end{array}
            \right] \\ 
            &\triangleq \left[\begin{array}{cc}
                    Q + A^TP_i + P_iA & P_iB \\
                    B^TP_i & R
            \end{array}\right].            
    \end{split}
    \end{equation*}
    \item (Policy improvement) \label{standard_PI_PI}Get the improved policy by
    \begin{equation}\label{RPI_PI}
        K_{i+1} = G^{-1}_{i,22}G_{i,21}.
    \end{equation}
    \item Set $i\gets i+1$ and go back to Step \ref{standard_PI_PE}.
\end{enumerate}
\end{proc}
\begin{theorem}[\cite{1098829}]\label{theorem_standard_PI}
    In Procedure \ref{procedure_policy_itration} we have:
    \begin{enumerate}[label=\roman*)]
        \item $A-BK_i$ is Hurwitz for all $i=1,2,\cdots$.  \label{Hurwitz_matrix}
        \item $P_1\geq P_2 \geq P_3\geq \cdots \geq P^*$.\label{PI_monotone}
        \item $\lim_{i\rightarrow\infty}P_i=P^*$, $\lim_{i\rightarrow\infty}K_i = K^*$.\label{PI_converge}
    \end{enumerate}
\end{theorem}
\begin{remark}\label{why_name_policy_evaluation}
    Step \ref{standard_PI_PE} in Procedure \ref{procedure_policy_itration} is named policy evaluation because $P_i$ is positive semidefinite, and the cost induced by control law $u_i=-K_ix$ when applied to system \eqref{LTI_sys} is $J(x_0, u_i) = x_0^TP_ix_0$.
\end{remark}
In Procedure \ref{procedure_policy_itration}, the exact knowledge of $A$ and $B$ is required, as the solution to \eqref{RPI_PE} relies upon $A$ and $B$. So, Kleinman's original policy iteration algorithm is model-based. However, in practice, very often we only have access to incomplete information required to solve the problem. In other words, each policy evaluation step will result in inaccurate estimation. Thus we are interested in studying the following problem.
\begin{problem}
If $G_i$ is replaced by an approximated matrix $\hat{G}_i$, will the conclusions in Theorem \ref{theorem_standard_PI} still hold?
\end{problem}
\begin{remark}\label{sources_of_errors}
The difference between $\hat{G}_i$ and $G_i$ can be attributed to errors from various sources, for example: the noises in RL and ADP algorithms; the estimation errors of $A$ and $B$ in indirect adaptive control and system identification \cite{mareels2012adaptive}; the residual caused by an early termination of the iteration to numerically solve ARE (\ref{ARE}), in order to save computational efforts \cite{hylla2011extension}; approximate values of $Q$ and $R$ in inverse optimal control, due to the absence of exact knowledge of the cost function \cite{6880317}.
\end{remark}

\section{Robustness Analysis of Policy Iteration}\label{section_robust_analysis}
In this section, we show that as long as the discrepancy between $G_i$ and $\hat{G}_i$ is small, the discrepancies between the performance value of the generated policies using $\hat{G}_i$ and that of the optimal policy will also be small in the limit.

We firstly reinterpret Procedure \ref{procedure_policy_itration} as a discrete-time nonlinear system. To this end, suppose
\begin{equation}\label{start_with_initial_matrix_standard_PI}
    K_1 = R^{-1}B^TP_0,
\end{equation}
where $P_0\in\mathbb{S}^{n}$ is chosen such that $K_1$ is stabilizing. Such a $P_0$ always exists under the above controllability and observability assumptions. The sequence $\{P_i\}_{i=0}^{\infty}$ generated by Procedure \ref{procedure_policy_itration} with condition \eqref{start_with_initial_matrix_standard_PI} can be regarded as the state, and the iteration index $i$ can be regarded as time. Then Procedure \ref{procedure_policy_itration} is a discrete-time nonlinear system and $P^*$ is an equilibrium by Theorem \ref{theorem_standard_PI}. The next lemma states that $P^*$ is actually a locally exponentially stable equilibrium.
\begin{lemma}\label{lemma_local_exponentail_stable}
    For any $\sigma<1$, there exists a $\delta_0(\sigma)>0$, such that for any $P_i\in\mathcal{B}_{\delta_0}(P^*)$, $\mathcal{A}(P_i)$ is Hurwitz and
    $$\Vert P_{i+1}-P^* \Vert_F \leq \sigma\Vert P_i-P^*\Vert_F.$$
\end{lemma}
\begin{remark}
Lemma \ref{lemma_local_exponentail_stable} is inspired by the fact that Procedure \ref{procedure_policy_itration} is quadratically convergent, which was observed by Kleinman in \cite[Remark 3)]{1098829}.
\end{remark}
Now consider the policy iteration in the presence of errors.
\begin{proc}[Robust Policy Iteration]\label{procedure_robust_policy_iteration}
\ \par
 \begin{enumerate}[label=\arabic*)]
        \item Choose a stabilizing control gain $\hat{K}_1$, and let $i=1$.
        \item\label{inexact_PI_PE} (Inexact policy evaluation) Obtain $\hat{G}_i = \tilde{G}_i + \Delta G_i$, where $\Delta G_i\in\mathbb{S}^{m+n}$ is a disturbance,
        $$\tilde{G}_i = \left[\begin{array}{cc}
        Q + A^T\tilde{P}_i + \tilde{P}_iA & \tilde{P}_iB \\
        B^T\tilde{P}_i & R
        \end{array}\right],$$
        and $\tilde{P}_i\in\mathbb{S}^n$ satisfy
        \begin{equation}\label{eRPI_PE}
        \mathcal{H}(\tilde{G}_i,\hat{K}_i)=0.
        \end{equation}
        \item (Policy update) Construct a new control gain
        \begin{equation}\label{eRPI_PI}
            \hat{K}_{i+1} = \hat{G}^{-1}_{i,22}\hat{G}_{i,21}.
        \end{equation}
        \item Set $i\gets i+1$ and go back to Step \ref{inexact_PI_PE}.
    \end{enumerate}
\end{proc}
\begin{remark}
The requirement that $\hat{G}_i\in\mathbb{S}^{m+n}$ is not restrictive, since for any $X\in\mathbb{R}^{(n+m)\times(n+m)}$, $x^TXx = \frac{1}{2}x^T(X+X^T)x$, where $\frac{1}{2}(X+X^T)$ is symmetric.
\end{remark}
Suppose
\begin{equation}\label{start_with_initial_matrix_robust_PI}
    \hat{K}_1 = R^{-1}B^T\tilde{P}_0
\end{equation}
is stabilizing, and $\Delta G_0=0$, where $\tilde{P}_0\in\mathbb{S}^{n}$. If $\hat{K}_i$ is stabilizing and $\hat{G}_{i,22}$ is invertible for all $i\in\mathbb{Z}_+$, $i>0$ (this is possible under certain conditions, see Appendix \ref{proof_local_ISS}), the sequence $\{\tilde{P}_i\}_{i=0}^{\infty}$ generated by Procedure \ref{procedure_robust_policy_iteration} with condition \eqref{start_with_initial_matrix_robust_PI} can be viewed as the state, and $\{\Delta G_i\}_{i=0}^{\infty}$ is the disturbance input. With this in mind, the next theorem shows that the discrete-time system $\tilde{P}_i$ as defined by the Procedure 2 above is locally input-to-state stable when $\Delta G$ is considered as the input; see \cite{JIANG2001857}, \cite[Theorem 2.1.]{JIANG20042129}.
\begin{theorem}\label{theorem_local_ISS}
For $\sigma$ and its associated $\delta_0$ in Lemma \ref{lemma_local_exponentail_stable}, there exists $\delta_1(\delta_0)>0$, such that if $\Vert\Delta G\Vert_\infty<\delta_1$, $\tilde{P}_0\in\mathcal{B}_{\delta_0}(P^*)$,
\begin{enumerate}[label=(\roman*)]
    \item\label{theorem_local_ISS_item_well_defined} $\hat{G}_{i,22}$ is invertible and $\hat{K}_i$ is stabilizing, $\forall i\in\mathbb{Z}_+$, $i>0$;
    \item\label{theorem_local_ISS_item_local_ISS} the following local input-to-state stability holds:
    $$\Vert \tilde{P}_i - P^*\Vert_F\leq \beta(\Vert\tilde{P}_0-P^*\Vert_F,i) + \gamma(\Vert \Delta G\Vert_\infty),$$
    where
    $$\beta(y,i) = \sigma^iy,\qquad \gamma(y) = \frac{c_3}{1-\sigma}y,\qquad y\in\mathbb{R},$$
    and $c_3(\delta_0)>0$.
    \item\label{theorem_local_ISS__item_converging_input_converging_state} $\lim_{i\rightarrow\infty} \Vert\Delta G_i\Vert_F = 0$ implies $\lim_{i\rightarrow\infty} \Vert \tilde{P}_i-P^* \Vert_F=0$.
\end{enumerate}
\end{theorem}
\begin{remark}
Since $\sigma<1$, $\beta\in\mathcal{K}\mathcal{L}$ and $\gamma(\cdot)\in\mathcal{K}$.
\end{remark}
Intuitively, Theorem \ref{theorem_local_ISS} implies that in Procedure \ref{procedure_robust_policy_iteration}, if $\tilde{P}_0$ is near $P^*$ (thus $\hat{K}_1$ is near $K^*$), and the disturbance input $\Delta G$ is bounded and not too large, then the cost of the generated control policy $\hat{K}_i$ is also bounded, and will ultimately be no larger than a constant proportional to the $l^\infty$-norm of the disturbance. The smaller the disturbance is, the better the ultimately generated policy is. In other words, the algorithm described in Procedure \ref{procedure_robust_policy_iteration} is not sensitive to small disturbances when the initial condition is in a neighbourhood of the optimal solution. Based on Theorem \ref{theorem_local_ISS}, the following corollary is derived. 
\begin{corollary}\label{corollary_global}
For any given stabilizing control gain $\hat{K}_1$ and any $\epsilon>0$, if $Q>0$, there exists $\delta_2(\epsilon,\hat{K}_1)>0$, $M_0(\delta_2)>0$, such that as long as
$\Vert \Delta G \Vert_\infty < \delta_2$, $\hat{G}_{i,22}$ is invertible, $\hat{K}_i$ is stabilizing, $\Vert \tilde{P}_i\Vert_F<M_0$, $\forall i\in\mathbb{Z}_+$, $i>0$ and
$$\limsup_{i\rightarrow\infty} \Vert \tilde{P}_i-P^* \Vert_F<\epsilon.$$
If in addition $\lim_{i\rightarrow\infty} \Vert\Delta G_i\Vert_F = 0$, then $\lim_{i\rightarrow\infty} \Vert \tilde{P}_i-P^* \Vert_F=0$.
\end{corollary}
\begin{remark}
In Corollary \ref{corollary_global}, $\hat{K}_1$ can be any stabilizing control gain, which is different from that of Theorem \ref{theorem_local_ISS}. When there is no disturbance, Corollary \ref{corollary_global} implies the convergence result of Procedure \ref{procedure_policy_itration} in \cite[Theorem]{1098829} (i.e. Theorem \ref{theorem_standard_PI} in this paper).
\end{remark}
\section{Data-driven Policy Iteration with Unkonwn Bounded Disturbances}\label{section_datadriven_PI}
In this section, we use Corollary \ref{corollary_global} to show that the Algorithm 1 in \cite{JIANG20122699}, a data-driven policy iteration method for the LQR problem, is robust to noisy data caused by small unknown additive external disturbances in the system dynamics.

Assume $A$, $B$ are unknown, and $Q$, $R$ are given. For a trajectory $\{x(t),u(t)\}$ of system \eqref{LTI_sys} on interval $[0,M\Delta t]$, where $M\in\mathbb{Z}_+$, $\Delta t\in\mathbb{R}_+$, consider the following equation 
\begin{equation}\label{ADP_LTI}
    \left[\begin{array}{cc}
         \svec(P_i)  \\
         \vect(K_{i+1}) 
    \end{array}\right] = \Theta(K_i)^\dagger\Xi(K_i),
\end{equation}
where
\begin{equation*}
    \begin{split}
        \Theta(K_i) &= \left[\delta_{xx}, -2I_{xx}(I_n\otimes K_i^TR)-2I_{xu}(I_n\otimes R)\right],  \\
        \Xi(K_i) &= -I_{xx}\vect(Q + K_i^TRK_i),   \\
        \delta_{xx} &= \left[\mathcal{U}_0, \mathcal{U}_1, \cdots, \mathcal{U}_{M-1}\right]^T, \ \mathcal{U}_j = \tilde{x}(t_{j+1})-\tilde{x}(t_j),  \\
        I_{xx} &= \left[\mathcal{V}_0, \mathcal{V}_1, \cdots, \mathcal{V}_{M-1}\right]^T,\ \mathcal{V}_j = \int_{t_j}^{t_{j+1}} x\otimes x \dd{\tau}, \\
        I_{xu} &= \left[\mathcal{W}_0, \mathcal{W}_1, \cdots, \mathcal{W}_{M-1}\right]^T,\ \mathcal{W}_j = \int_{t_j}^{t_{j+1}} x\otimes u \dd{\tau},
    \end{split}
\end{equation*}
and $t_j=\Delta t j$, $j=0,\cdots,M$. Note that only input/state data is involved in the matrices $\Theta(K_i)$ and $\Xi(K_i)$. It has been shown in \cite[Lemma 6, Theorem 7]{JIANG20122699} that, if
\begin{equation}\label{ADP_LTI_assum}
    \rank([I_{xx},I_{xu}]) = n_0 \triangleq \frac{n(n+1)}{2} + mn,
\end{equation}
then starting with a stabilizing initial control gain $K_1$, we can obtain the sequences of $\{P_i\}_{i=1}^\infty$ and $\{K_i\}_{i=1}^\infty$ in Procedure \ref{procedure_policy_itration} by iterating equation (\ref{ADP_LTI}), instead of iterating between Step \ref{standard_PI_PE} and Step \ref{standard_PI_PI}.

Now suppose system \eqref{LTI_sys} is corrupted by external disturbance,
\begin{equation}\label{LTI_sys_noise}
    \dot{\hat{x}} = A\hat{x} + B\hat{u} + w,\quad \hat{x}(0) = \hat{x}_0,
\end{equation}
where $w\in\mathbb{R}^n$ is the disturbance and unknown. Replacing $x$ and $u$ in \eqref{ADP_LTI} by $\hat{x}$ and $\hat{u}$ yields
\begin{equation}\label{ADP_LTI_noise}
    \left[\begin{array}{cc}
         \svec(\hat{P}_i)  \\
         \vect(\hat{K}_{i+1}) 
    \end{array}\right] = \hat{\Theta}(\hat{K_i})^\dagger\hat{\Xi}(\hat{K}_i),
\end{equation}
where
\begin{equation*}
    \begin{split}
        \hat{\Theta}(\hat{K}_i) &= \left[\delta_{\hat{x}\hat{x}}, -2I_{\hat{x}\hat{x}}(I_n\otimes \hat{K}_i^TR)-2I_{\hat{x}\hat{u}}(I_n\otimes R)\right],  \\
        \hat{\Xi}(\hat{K}_i) &= -I_{\hat{x}\hat{x}}\vect(Q + \hat{K}_i^TR\hat{K}_i).
    \end{split}
\end{equation*}
Using Corollary \ref{corollary_global}, we can show that approximate optimal control gain can still be obtained by iterating \eqref{ADP_LTI_noise}.
\begin{theorem}\label{theorem_data_driven_robut_PI}
Given $M\in\mathbb{Z}_+$, $\Delta t\in\mathbb{R}_+$, starting from any stabilizing control gain $\hat{K}_1$, for any $\epsilon>0$, there exists a $c(\epsilon, \hat{K}_1)>0$, such that if $\rank([I_{\hat{x}\hat{x}},I_{\hat{x}\hat{u}}]) = n_0$ for any $\Vert w\Vert_\infty<c$, then $\hat{K}_i$ given by iterating \eqref{ADP_LTI_noise} is stabilizing for all $i\in\mathbb{Z}_+$, $i>0$ and $\limsup_{i\rightarrow\infty} \Vert \tilde{P}_i-P^* \Vert_F<\epsilon.$
\end{theorem}
\begin{remark}
The rank condition in Theorem \ref{theorem_data_driven_robut_PI} is in the spirit of persistent excitation in adaptive control \cite{mareels2012adaptive}, which is common in the adaptive dynamic programming literature (see \cite{jiang2017robust,Frank-book1,doi:10.1137/18M1214147,pang2019adaptive,Pang2020policy,8619347,JIANG20122699}). If condition \eqref{ADP_LTI_assum} is satisfied, then by continuity, there always exists a $c>0$ small enough such that the rank condition in Theorem \ref{theorem_data_driven_robut_PI} holds. Usually sum of sinusoidal signals with different frequencies can be added to the input to satisfy condition \eqref{ADP_LTI_assum}, see \cite{mareels2012adaptive,JIANG20122699}.
\end{remark}
\section{Numerical Example}\label{section_simulation}
In this section, we apply \eqref{ADP_LTI_noise} to a disturbed model of constant level continuous stirred tank reactor considered in \cite[Example 2]{doi:10.1080/002071797223343}
$$\dot{x} = \left[\begin{array}{cc}
    -21 & -20 \\
    9 & 8
\end{array}\right]x + u + \xi\sum_{j=1}^{50} sin(\eta_j t),$$
where $\{\eta_j\}_{j=1}^{50}$ is a set of frequencies uniformly sampled from interval $[-100,100]$, and $\xi\in\mathbb{R}_+$ controls the scale of the disturbance. In the simulation, we set $M=140$, $\Delta t = 0.1$, and $u = 0.2\sum_{j=1}^{100} sin(\psi_j t),$
where $\{\psi_j\}_{j=1}^{100}$ is a set of frequencies uniformly sampled from interval $[-500,500]$. \eqref{ADP_LTI_noise} is applied for different choices of $\hat{K}_1$ and $\xi$. The results are summarized in Fig. \ref{comparison}. Comparing Fig.\ref{comparison}-(a) with Fig.\ref{comparison}-(b) and Fig.\ref{comparison}-(c) with Fig.\ref{comparison}-(d) implies that policy iteration is robust to small bounded unknown disturbances, no matter whether the initial control gain is near the optimal gain or not. Comparing Fig.\ref{comparison}-(a) with Fig.\ref{comparison}-(c) and Fig.\ref{comparison}-(b) with Fig.\ref{comparison}-(d) implies that the degree of suboptimality is proportional to the scale of the disturbance. These observations are consistent with the theories we developed in this paper.
\begin{figure}[!htb]
    \centering
    \includegraphics[width=\linewidth]{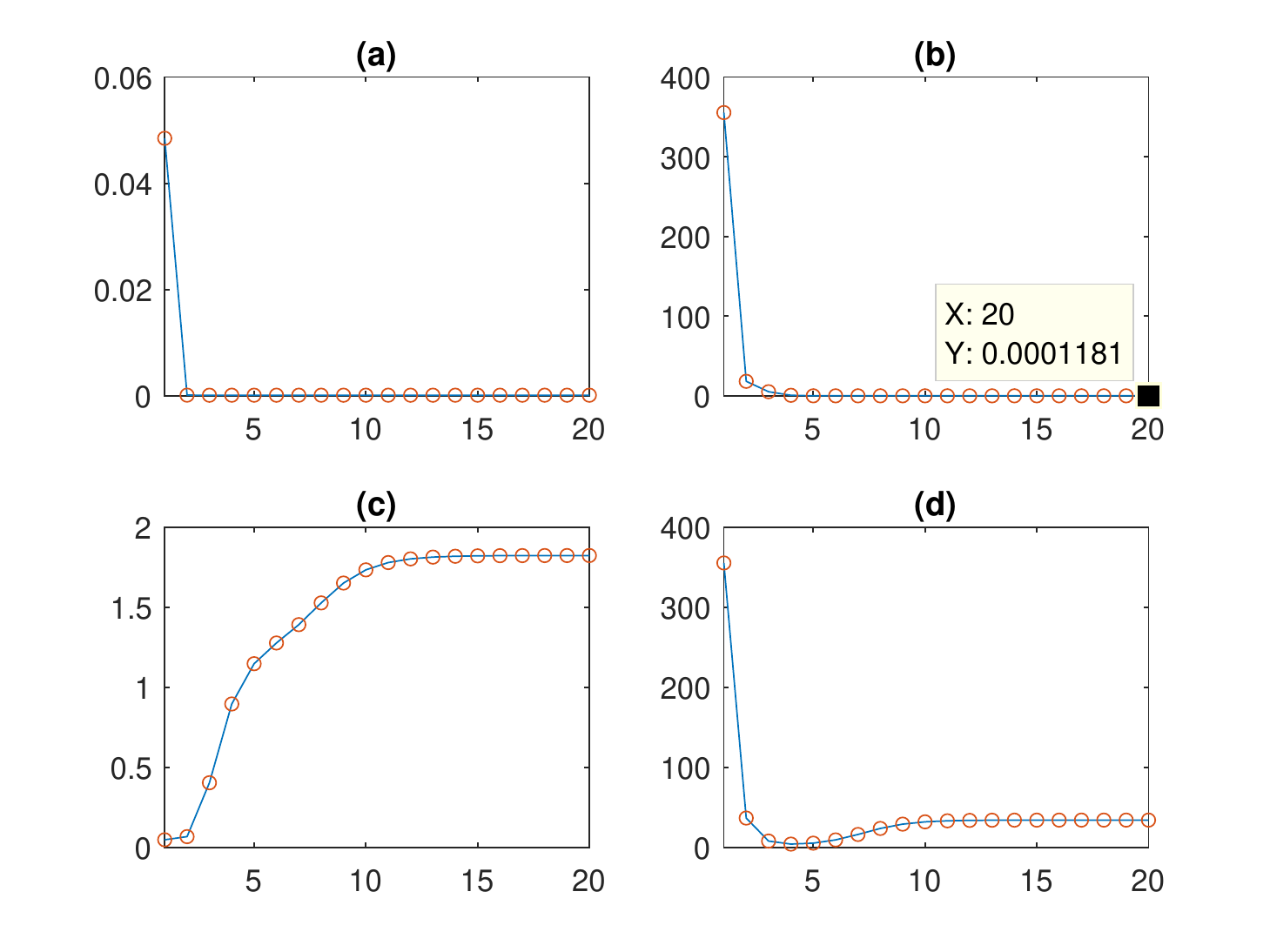}
    \caption{$\Vert \tilde{P}_i-P^*\Vert_F$ (y-axis) versus iteration index (x-axis). (a) and (c) share the same $\hat{K}_1$, which is near to $K^*$. (b) and (d) share the same $\hat{K}_1$, which is far away from $K^*$. (a) and (b) share the same $\xi = 0.01$. (c) and (d) share the same $\xi = 0.5$. All of the $\hat{K}_i$ generated in these experiments are stabilizing.}\label{comparison}
\end{figure}

\section{Conclusion}\label{section_conclusion}
This paper has proposed a robust policy iteration method for the continuous-time LQR problem in the presence of errors at the performance evaluation and performance improvement steps. Some robustness results in the sense of Sontag's input-to-state stability are given. More precisely, it is shown that starting from any stabilizing control policy, the solutions generated by policy iteration are guaranteed close to the optimal solution, as long as the errors are bounded and small. In addition, when the system dynamics are not exactly known, a novel off-policy data-driven policy iteration method is proposed with guaranteed robustness to noisy data collected from the system trajectories in the presence of unknown bounded disturbances.






\appendix
For $X\in\mathbb{R}^{n\times n}$, $Y\in\mathbb{R}^{n\times n}$, define Lyapunov operator
$$\mathcal{L}_{X}(Y)=X^TY+YX.$$
Some useful properties of $\mathcal{L}_X(\cdot)$ are provided below.
\begin{lemma}\label{Lyapunov_properties}
Suppose $X\in\mathbb{R}^{n\times n}$ is Hurwitz.
\begin{enumerate}[label=(\roman*)]
    \item\label{Lyapunov_inverse} 
    The inverse operator $\mathcal{L}_X^{-1}(\cdot)$ exists on $\mathbb{R}^{n\times n}$, and
    $$\Vert \mathcal{L}_X^{-1}\Vert = \Vert \mathcal{P}(X)^{-1}\Vert_2,$$ where
    $$\Vert \mathcal{L}_{X}\Vert = \sup_{Z\in\mathbb{R}^{n\times n},Z\neq 0}\frac{\Vert\mathcal{L}_{X}(Z)\Vert_F}{\Vert Z\Vert_F},$$
    and
    $$\mathcal{P}(X) = I_n\otimes X^T + X^T\otimes I_n.$$
    \item\label{Lyapunov_solution} For each $Z\in\mathbb{R}^{n\times n}$, if $Y$ is the solution of equation $\mathcal{L}_X(Y)=-Z$, then
    $$Y=\int_0^{\infty} e^{X^Tt}Ze^{Xt}\dd t.$$
    \item\label{Lyapunov_perturbation} If in addition to the condition in \ref{Lyapunov_solution}, $\mathcal{L}_{X+\Delta X}(Y+\Delta Y) = -(Z+\Delta Z)$, for some $\Delta X, \Delta Y, \Delta Z\in\mathbb{R}^{n\times n}$, then
    $$\Vert \Delta Y\Vert_2 \leq \left(\Vert\Delta Z\Vert_2 + 2\Vert \Delta X\Vert_2\Vert Y+\Delta Y\Vert_2\right) \Vert H\Vert_2,$$
    where $H = \mathcal{L}^{-1}_X(-I_n)$.
    \item\label{lyapunov_perturbation_small} If in addition to the conditions in \ref{Lyapunov_perturbation}, $Z$ is nonzero. Let $\gamma>0$ be such that
    \begin{equation*}
    \begin{split}
        \Vert \Delta X\Vert_F \leq \gamma \Vert X\Vert_F,\quad
        \Vert \Delta Z\Vert_F &\leq \gamma \Vert Z\Vert_F \\
        \gamma\Vert X\Vert_F \Vert \mathcal{P}(X)^{-1}\Vert_2 &\leq 1/4
    \end{split}
    \end{equation*}
    then
    $$\Vert \Delta Y\Vert_F\leq 8\gamma\Vert X\Vert_F\Vert \mathcal{P}(X)^{-1}\Vert_2\Vert Y\Vert_F.$$
    \item\label{Lyapunov_monotone} $\mathcal{L}_{X}(Y_1)\leq \mathcal{L}_{X}(Y_2) \implies Y_1\geq Y_2$, where $Y_1,Y_2\in\mathbb{S}^n$.
\end{enumerate}
\end{lemma}
\begin{proof}
\ref{Lyapunov_solution} is direct consequence of \cite[Theorem 2.2.2.]{hylla2011extension}. \ref{Lyapunov_perturbation} comes from the equation (2.11) in the proof of \cite[Theorem 2.1]{doi:10.1137/0326018}. \ref{Lyapunov_inverse} and \ref{lyapunov_perturbation_small} come from \cite[Section 4]{1102170}. For \ref{Lyapunov_monotone}, let $x$ be the solution of system $\dot{x}=Xx$, $x(0)=x_0$. Then,
\begin{equation*}
    x^T\mathcal{L}_{X}(Y_1)x=\diff{x^TY_1x}{t}\leq\diff{x^TY_2x}{t}=x^T\mathcal{L}_{X}(Y_2)x.
\end{equation*}
Integrating both sides from $0$ to $\infty$, we get $x_0^TY_1x_0\geq x_0^TY_2x_0$. $x_0$ is arbitrary, thus $Y_1\geq Y_2$.
\end{proof}

\subsection{Proof of Lemma \ref{lemma_local_exponentail_stable}}\label{appendix_locally_exponential_stability}
By \eqref{RPI_PE}, the sequence $\{P_i\}_{i=0}^{\infty}$ satisfies 
\begin{equation}\label{Kleinman_one_step}
    P_{i+1} = \mathcal{L}_{\mathcal{A}(P_i)}^{-1}\left(- Q - P_i^TBR^{-1}B^TP_i\right),
\end{equation}
where
$$\mathcal{A}(P_i) = A - BR^{-1}B^TP_i.$$
Since $\mathcal{A}(P^*)$ is Hurwitz, by continuity there always exists a $\bar{\delta}_0>0$, such that $\mathcal{A}(P_i)$ is Hurwitz for all $P_i\in\mathcal{\bar{B}}_{\bar{\delta}_0}(P^*)$. Suppose $P_i\in\mathcal{\bar{B}}_{\bar{\delta}_0}(P^*)$. Subtracting $P^*BR^{-1}B^TP_i + P_i^TBR^{-1}B^TP^*$ from the both sides of the ARE \eqref{ARE} yields
\begin{equation}\label{ARE_reformulation}
\begin{split}
    \mathcal{L}_{\mathcal{A}(P_i)}(P^*) &= - Q - P_i^TBR^{-1}B^TP_i  \\
    &+ (P_i^T-P^*)BR^{-1}B^T(P_i-P^*).
\end{split}
\end{equation}
Subtracting \eqref{ARE_reformulation} from \eqref{Kleinman_one_step}, we have
$$P_{i+1}-P^* = \mathcal{L}^{-1}_{\mathcal{A}(P_i)}(-(P_i^T-P^*)BR^{-1}B^T(P_i-P^*)).$$
Taking norm on both sides of above equation, and using \ref{Lyapunov_inverse} in Lemma \ref{Lyapunov_properties}, we have
\begin{equation*}
    \begin{split}
        \Vert P_{i+1}-P^* \Vert_F &\leq \Vert\mathcal{L}^{-1}_{\mathcal{A}(P_i)}\Vert\Vert BR^{-1}B^T\Vert_F\Vert P_i-P^*\Vert_F^2 \\
        & = \Vert \mathcal{P}(\mathcal{A}(P_i))^{-1}\Vert_2\Vert BR^{-1}B^T\Vert_F\Vert P_i-P^*\Vert_F^2.
    \end{split}
\end{equation*}
By continuity of matrix norm and matrix inverse, there exist a $c_1>0$, such that
$$\Vert \mathcal{P}(\mathcal{A}(P_i))^{-1}\Vert_2\leq c_1,\quad \forall P_i\in\mathcal{\bar{B}}_{\bar{\delta}_0}(P^*).$$
So for any $0<\sigma<1$, there exists a $\bar{\delta}_0\geq\delta_0>0$, such that
$$c_1\Vert BR^{-1}B^T\Vert_F\delta_0\leq \sigma,$$
which completes the proof.

\subsection{Proof of Theorem \ref{theorem_local_ISS}}\label{proof_local_ISS}
Before proving Theorem \ref{theorem_local_ISS}, we firstly prove some auxiliary lemmas. The Procedure \ref{procedure_robust_policy_iteration} will exhibit a singularity, if $\hat{G}_{i,22}$ in \eqref{eRPI_PI} is singular, or the cost \eqref{cost_deterinistic} of $\hat{K}_{i+1}$ is infinity. The following lemma shows that if $\Delta G_i$ is small, no singularity will occur. Let $\bar{\delta}_0$ be the one defined in Appendix \ref{appendix_locally_exponential_stability}, then $\delta_0\leq\bar{\delta}_0$. For convenience, define $\hat{A}_i = A-B\hat{K}_i$.
\begin{lemma}\label{lemma_stability_ISS}
For any $\tilde{P}_i\in\mathcal{B}_{\delta_0}(P^*)$, there exists a $d(\delta_0)>0$, independent of $\tilde{P}_i$, such that $\hat{K}_{i+1}$ is stabilizing and $\hat{G}_{i,22}$ is invertible, if $\Vert\Delta G_i\Vert_F< d$.
\end{lemma}
\begin{proof}
Since $\mathcal{\bar{B}}_{\bar{\delta}_0}(P^*)$ is compact and $\mathcal{A}(\cdot)$ is a continuous function, set $\mathcal{S} = \{\mathcal{A}(\tilde{P}_i)\vert \tilde{P}_i\in\mathcal{\bar{B}}_{\bar{\delta}_0}(P^*)\}$ is also compact. By continuity, for each $X\in\mathcal{S}$, there exists a $r(X)>0$ such that any $Y\in\mathcal{B}_{r(X)}(X)$ is Hurwitz. By Heine–Borel theorem, there exists a $\underline{r}>0$, such that each $Y\in\mathcal{B}_{\underline{r}}(X)$ is Hurwitz for all $X\in\mathcal{S}$. By continuity of matrix inversion, there exists $d_1>0$ such that $\hat{G}_{i,22}$ is invertible, if $\Vert\Delta G_i\Vert_F< d_1$. Note that in policy improvement step of Procedure \ref{procedure_policy_itration} (the policy update step in Procedure \ref{procedure_robust_policy_iteration}), the improved policy $\tilde{K}_{i+1}=\tilde{G}_{i,22}^{-1}\tilde{G}_{i,21}$ (the updated policy $\hat{K}_{i+1}$) is continuous function of $\tilde{G}_i$ ($\hat{G}_i$), and there exists a $d_2>0$, such that $\hat{A}_{i+1}\in\mathcal{B}_{\underline{r}}(\mathcal{A}(\tilde{P}_i))$ for all $\tilde{P}_i\in\mathcal{\bar{B}}_{\bar{\delta}_0}(P^*)$, if $\Vert\Delta G_i\Vert_F< d_2$. Thus $\hat{K}_{i+1}$ is stabilizing. Setting $d=\min(d_1,d_2)$ completes the proof.
\end{proof}
By \eqref{eRPI_PE} and Lemma \ref{lemma_stability_ISS}, if $\Vert\Delta G_i\Vert_F< d$, then the sequence $\{\tilde{P}_i\}_{i=0}^{\infty}$ satisfies 
\begin{equation}\label{robust_Kleinman_one_step}
    \tilde{P}_{i+1} = \mathcal{L}_{\mathcal{A}(\tilde{P}_i)}^{-1}\left(- Q - \tilde{P}^T_iBR^{-1}B^T\tilde{P}_i\right) + \mathcal{E}_i,
\end{equation}
where
\begin{equation*}
    \begin{split}
        \mathcal{E}_i &= \mathcal{L}_{\hat{A}_{i+1}}^{-1}\left(-Q-\hat{K}^T_{i+1}R\hat{K}_{i+1}\right) \\
        &- \mathcal{L}_{\mathcal{A}(\tilde{P}_i)}^{-1}\left(- Q - \tilde{P}^T_iBR^{-1}B^T\tilde{P}_i\right).
    \end{split}
\end{equation*}
The following lemma gives an upper bound of $\Vert \mathcal{E}_i\Vert_F$ in terms of $\Vert\Delta G_i\Vert_F$.
\begin{lemma}\label{lemma_total_error_bounded_ISS}
For any $\tilde{P}_i\in\mathcal{B}_{\delta_0}(P^*)$ and any $c_2>0$, there exists a $0<\delta_1^1(\delta_0,c_2)\leq d$, independent of $\tilde{P}_i$, where $d$ is defined in Lemma \ref{lemma_stability_ISS}, such that
$$\Vert \mathcal{E}_i\Vert_F\leq c_3\Vert \Delta G_i\Vert_F<c_2,$$ 
if $\Vert \Delta G_i\Vert_F<\delta_1^1$, where $c_3(\delta_0)>0$.
\end{lemma}
\begin{proof}
In view of \ref{Lyapunov_perturbation} in Lemma \ref{Lyapunov_properties}, set 
\begin{equation*}
    \begin{split}
        &X = \mathcal{A}(\tilde{P}_i),\ X+\Delta X = \hat{A}_{i+1},    \\
        &Z = -Q - \tilde{P}^T_iBR^{-1}B^T\tilde{P}_i,\ Z + \Delta Z = -Q-\hat{K}^T_{i+1}R\hat{K}_{i+1},  \\
        &Y = \mathcal{L}_{\mathcal{A}(\tilde{P}_i)}^{-1}\left(- Q - \tilde{P}^T_iBR^{-1}B^T\tilde{P}_i\right),  \\
        &Y+\Delta Y=\mathcal{L}_{\hat{A}_{i+1}}^{-1}\left(-Q-\hat{K}^T_{i+1}R\hat{K}_{i+1}\right).
    \end{split}
\end{equation*}
Then by \eqref{robust_Kleinman_one_step}, $\Delta Y = \mathcal{E}_i$.

Define
\begin{equation*}
    \begin{split}
    \underline{\gamma_1} &= \max\left\{\Vert Y\Vert_F\vert \tilde{P}_i\in\mathcal{\bar{B}}_{\bar{\delta}_0}(P^*)\right\},  \\
    \underline{\gamma_2} &= \max\left\{8\Vert X\Vert_F\Vert \mathcal{P}(X)^{-1}\Vert_2\vert \tilde{P}_i\in\mathcal{\bar{B}}_{\bar{\delta}_0}(P^*)\right\} \\
    \underline{\gamma_3} &= \min\{\Vert X\Vert_F\vert \tilde{P}_i\in\mathcal{\bar{B}}_{\bar{\delta}_0}(P^*)\},\\ \underline{\gamma_4} &= \min\{\Vert Z \Vert_F\vert \tilde{P}_i\in\mathcal{\bar{B}}_{\bar{\delta}_0}(P^*)\}, \\
    \underline{\gamma_5} &= \min\left\{c_2(\underline{\gamma_1}\cdot\underline{\gamma_2})^{-1},\ 2\underline{\gamma_2}^{-1} \right\},
\end{split}
\end{equation*}
By continuity, $\underline{\gamma_i}$, $i=1,\cdots,5$ are all well-defined positive and finite constants. 

For all $\tilde{P}_i\in\mathcal{B}_{\delta_0}(P^*)$, it can be checked that $\Vert\Delta X\Vert_F\leq c_4\Vert \Delta G_i\Vert_F$, $\Vert\Delta Z\Vert_F\leq c_5\Vert \Delta G_i\Vert_F$ for some $c_4(\delta_0,d)>0$ and $c_5(\delta_0,d)>0$, then we have
\begin{equation*}
    \max\left(\Vert\Delta X\Vert_F/\Vert X\Vert_F, \Vert\Delta Z\Vert_F/\Vert Z\Vert_F\right)\leq c_6 \Vert \Delta G_i\Vert_F,
\end{equation*}
where $c_6 = \max(\underline{\gamma_3}^{-1}c_4, \underline{\gamma_4}^{-1}c_5)$. Thus, there exists $0<\delta_1^1\leq d$, such that $c_6\delta_1^1<\underline{\gamma_5}$. Then by setting $\gamma = c_6\Vert \Delta G_i\Vert_F$, one can check that the conditions in \ref{lyapunov_perturbation_small} in Lemma \ref{Lyapunov_properties} are satisfied for all $\Vert \Delta G_i\Vert_F<\delta_1^1$ and $\tilde{P}_i\in\mathcal{B}_{\delta_0}(P^*)$, which completes the proof, with $c_3 = \underline{\gamma_1}c_6\underline{\gamma_2}$.
\end{proof}
Now we are ready to prove Theorem \ref{theorem_local_ISS}.
\begin{proof}[Proof of Theorem \ref{theorem_local_ISS}]
Let $c_2=(1-\sigma)\delta_0$ in Lemma \ref{lemma_total_error_bounded_ISS}, and $\delta_1$ be equal to the $\delta_1^1$ associated with $c_2$. Then \ref{theorem_local_ISS_item_well_defined} in Theorem \ref{theorem_local_ISS} is immediately obtained. 

For any $i\in\mathbb{Z}_+$, if $\tilde{P}_i\in\mathcal{B}_{\delta_0}(P^*)$, then
\begin{align}
    \Vert \tilde{P}_{i+1} - P^*\Vert_F&\leq \Vert \mathcal{E}_i\Vert_F +\nonumber \\
    &\left\Vert \mathcal{L}^{-1}_{\mathcal{A}(\tilde{P}_i)}(Q+\tilde{P}^T_iBR^{-1}B^T\tilde{P}_i) - P^* \right\Vert_F \nonumber \\
    &\leq \sigma\Vert \tilde{P}_i - P^*\Vert_F + c_3\Vert\Delta G_i\Vert_F \label{proof_theorem_local_ISS_inequality_1}\\
    &\leq  \sigma\Vert \tilde{P}_i - P^*\Vert_F + c_3\Vert\Delta G\Vert_\infty \label{proof_theorem_local_ISS_inequality_2}\\
    &<\sigma\delta_0 + c_3\delta_1<\sigma\delta_0 + c_2 =\delta_0, \label{proof_theorem_local_ISS_inequality_3}
\end{align}
where \eqref{proof_theorem_local_ISS_inequality_1} and \eqref{proof_theorem_local_ISS_inequality_3} are due to Lemmas \ref{lemma_local_exponentail_stable} and \ref{lemma_total_error_bounded_ISS}. By induction, \eqref{proof_theorem_local_ISS_inequality_1} to \eqref{proof_theorem_local_ISS_inequality_3} hold for all $i\in\mathbb{Z}_+$, thus by \eqref{proof_theorem_local_ISS_inequality_2},
\begin{equation*}
    \begin{split}
        &\Vert \tilde{P}_{i} - P^*\Vert_F \leq \sigma^2\Vert \tilde{P}_{i-2}-P^*\Vert_F + (\sigma + 1)c_3\Vert\Delta G\Vert_\infty\\
        &\leq \cdots \leq \sigma^{i}\Vert \tilde{P}_{0}-P^*\Vert_F + (1+ \cdots + \sigma^{i-1})c_3\Vert\Delta G\Vert_\infty\\
        &<\sigma^{i}\Vert \tilde{P}_{0}-P^*\Vert_F + \frac{c_3}{1-\sigma}\Vert\Delta G\Vert_\infty,
    \end{split}
\end{equation*}
which proves \ref{theorem_local_ISS_item_local_ISS} in Theorem \ref{theorem_local_ISS}.

In terms of \ref{theorem_local_ISS__item_converging_input_converging_state} in Theorem \ref{theorem_local_ISS}, for any $\epsilon>0$, there exists a $i_1\in\mathbb{Z}_+$, such that $\sup\{\Vert\Delta G_i\Vert_F\}_{i=i_1}^\infty<\gamma^{-1}(\epsilon/2)$. Take $i_2\geq i_1$. For $k\geq i_2$, because $\tilde{P}_i$ is bounded, we have by \ref{theorem_local_ISS_item_local_ISS} in Lemma \ref{theorem_local_ISS},
\begin{equation*}
    \begin{split}
       \Vert \tilde{P}_{i}-P^*\Vert_F&\leq\beta(\Vert\tilde{P}_{i_2}-P^*\Vert_F,i-i_2) + \epsilon/2 \\
       &\leq \beta(c_7,i-i_2) + \epsilon/2.
    \end{split}
\end{equation*}
Since $\lim_{i\rightarrow\infty}\beta(c_7,i-i_2)=0$, there is a $i_3\geq i_2$ such that $\beta(c_7,i-i_2)<\epsilon/2$ for all $i\geq i_3$, which completes the proof.
\end{proof}
\subsection{Proof of Corollary \ref{corollary_global}}\label{proof_corollary_global}
Notice that all the conclusions of Corollary \ref{corollary_global} can be implied by Theorem \ref{theorem_local_ISS} if $\delta_2<\min(\gamma^{-1}(\epsilon),\delta_1)$, and $\tilde{P}_1\in\mathcal{B}_{\delta_0}(P^*)$ for Procedure \ref{procedure_robust_policy_iteration}. Thus the proof of Corollary \ref{corollary_global} reduces to the proof of the following lemma.
\begin{lemma}\label{lemma_converge_to_neighbourhood}
Given a stabilizing $\hat{K}_1$, there exists a $0<\delta_2<\min(\gamma^{-1}(\epsilon),\delta_1)$ and a $\bar{i}\in\mathbb{Z}_+$, such that $\tilde{P}_{\bar{i}}\in\mathcal{B}_{\delta_0}(P^*)$, as long as $\Vert \Delta G\Vert_\infty<\delta_2$.
\end{lemma}
The next two lemmas state that under certion conditions on $\Vert \Delta G_i\Vert_F$, each element in $\{\hat{K}_i\}_{i=1}^{\bar{i}}$ is stabilizing, each element in $\{\hat{G}_{i,22}\}_{i=1}^{\bar{i}}$ is invertible and $\{\tilde{P}_i\}_{i=1}^{\bar{i}}$ is bounded. For simplicity, in the following we assume $Q>I_n$ and $R>I_m$. All the proofs still work for any $Q>0$ and $R>0$, by suitable rescaling.
\begin{lemma}\label{lemma_stability}
If $\hat{K}_i$ is stabilizing, then $\hat{G}_{i,22}$ is nonsingular and $\hat{K}_{i+1}$ is stabilizing, as long as $\Vert \Delta G_i\Vert_F < a_i$, where
\begin{equation*}
     a_i=\left(m(\sqrt{n} + \Vert \hat{K}_{i}\Vert_2)^2+m(\sqrt{n} + \Vert \hat{K}_{i+1}\Vert_2)^2\right)^{-1}.
\end{equation*}
Furthermore,
\begin{equation}
    \Vert\hat{K}_{i+1}\Vert_F \leq 2\Vert R^{-1}\Vert_F(1 + \Vert B^T\tilde{P}_i\Vert_F).  \label{K_bound_by_P}
\end{equation}
\end{lemma}
\begin{proof}
By definition,
$$\Vert \tilde{G}_{i,22}^{-1}(\hat{G}_{i,22}-\tilde{G}_{i,22})\Vert_F < a_i\Vert \tilde{G}_{i,22}^{-1}\Vert_F.$$
Since $R>I_m$, the eigenvalues $\lambda_j(\tilde{G}_{i,22}^{-1})\in(0,1]$ for all $1\leq j\leq m$. Then by the fact that for any $X\in\mathbb{S}^m$
$$\Vert X\Vert_F = \Vert \Lambda_X\Vert_F,\quad \Lambda_X = \mathrm{diag}\{\lambda_1(X),\cdots,\lambda_m(X)\},$$
we have 
\begin{equation}\label{R_inverse_pertrubed_bound}
    \Vert \tilde{G}_{i,22}^{-1}(\hat{G}_{i,22}-\tilde{G}_{i,22})\Vert_F< a_i\sqrt{m} < 0.5.
\end{equation}
Thus by \cite[Section 5.8]{horn2012matrix}, $\hat{G}_{i,22}$ is invertible.

For any $x\in\mathbb{R}^{n}$ on the unit ball, define
$$\mathcal{X}_{\hat{K}_i} = \left[\begin{array}{c}
         I  \\
         -\hat{K}_i
    \end{array}\right]xx^T \left[\begin{array}{cc}
        I &  -\hat{K}_i^T
    \end{array}\right].$$
From \eqref{eRPI_PE} and \eqref{eRPI_PI} we have
$$x^T\mathcal{H}(\tilde{G}_i,\hat{K}_i)x = \trace(\tilde{G}_i\mathcal{X}_{\hat{K}_i}) = 0,$$
and
$$\trace(\hat{G}_i\mathcal{X}_{\hat{K}_{i+1}}) = \min_{K\in\mathbb{R}^{m\times n}} \trace(\hat{G}_i\mathcal{X}_K).$$
Then 
\begin{align}
    &\mathrm{tr}(\tilde{G}_i\mathcal{X}_{\hat{K}_{i+1}}) \leq \mathrm{tr}(\hat{G}_i\mathcal{X}_{\hat{K}_{i+1}}) + \Vert \Delta G_i\Vert_F \mathrm{tr}(\mathbf{1}\mathbf{1}^T\vert \mathcal{X}_{\hat{K}_{i+1}}\vert_{abs}) \nonumber\\   
    &\leq \mathrm{tr}(\hat{G}_i\mathcal{X}_{\hat{K}_{i}}) + \Vert \Delta G_i\Vert_F \mathbf{1}^T\vert \mathcal{X}_{\hat{K}_{i+1}}\vert_{abs}\mathbf{1}\nonumber \\
    &\leq \mathrm{tr}(\tilde{G}_i\mathcal{X}_{\hat{K}_{i}}) + \Vert \Delta G_i\Vert_F \mathbf{1}^T(\vert \mathcal{X}_{\hat{K}_{i}}\vert_{abs}+\vert \mathcal{X}_{\hat{K}_{i+1}}\vert_{abs})\mathbf{1} \nonumber\\
    &\leq \Vert \Delta G_i\Vert_F \mathbf{1}^T(\vert \mathcal{X}_{\hat{K}_{i}}\vert_{abs}+\vert \mathcal{X}_{\hat{K}_{i+1}}\vert_{abs})\mathbf{1}, \label{key_inequality_1}
\end{align}
where $\vert \mathcal{X}_{\hat{K}_{i}}\vert_{abs}$ denotes the matrix obtained from $\mathcal{X}_{\hat{K}_{i}}$ by taking the absolute value of each entry. Thus by (\ref{key_inequality_1}) and the definition of $\tilde{G}_i$, we have
\begin{equation}\label{key_inequality_2}
    x^T\mathcal{L}_{\hat{A}_{i+1}}(\tilde{P}_i)x + \epsilon_1 \leq 0
\end{equation}
where 
\begin{equation*}
\begin{split}
    \epsilon_1 &= x^T(Q+\hat{K}_{i+1}^TR\hat{K}_{i+1})x \\
    &- \Vert \Delta G_i\Vert_F \mathbf{1}^T(\vert\mathcal{X}_{\hat{K}_{i}}\vert_{abs}+\vert \mathcal{X}_{\hat{K}_{i+1}}\vert_{abs})\mathbf{1}.    
\end{split}
\end{equation*}
For any $x$ on the unit ball, $\vert\mathbf{1}^Tx\vert_{abs}\leq \sqrt{n}$. Similarly, for any $K\in\mathbb{R}^{m\times n}$, by the definition of induced matrix norm, $\vert\mathbf{1}^TKx\vert_{abs}\leq\Vert K\Vert_2 \sqrt{m}$. This implies
\begin{equation*}
    \left\vert\mathbf{1}^T\left[\begin{array}{c}
    I \\
    -K
    \end{array}\right]x\right\vert_{abs} = \left\vert\mathbf{1}^Tx - \mathbf{1}^TKx\right\vert_{abs} \leq \sqrt{m}(\sqrt{n} + \Vert K\Vert_2),
\end{equation*}
which means $\mathbf{1}^T\vert \mathcal{X}_K\vert_{abs}\mathbf{1}\leq m(\sqrt{n} + \Vert K\Vert_2)^2$. Thus
$$\Vert \Delta G_i\Vert_F \mathbf{1}^T(\vert \mathcal{X}_{\hat{K}_{i}}\vert_{abs}+\vert \mathcal{X}_{\hat{K}_{i+1}}\vert_{abs})\mathbf{1}<1.$$
Then $Q>I_n$ leads to
$$x^T\left(\hat{A}_{i+1}^T\tilde{P}_i+\tilde{P}_i\hat{A}_{i+1}\right)x<0$$
for all $x$ on the unit ball. So $\hat{K}_{i+1}$ is stabilizing by the Lyapunov criterion \cite[Lemma 12.2]{wonham1974linear}.

By definition, 
\begin{align}
        \Vert\hat{K}_{i+1}\Vert_F
        &\leq \Vert \hat{G}_{i,22}^{-1}\Vert_F(1 + \Vert B^T\tilde{P}_i\Vert_F)\nonumber \\
        &\leq \Vert R^{-1}\Vert_F(1-\Vert \tilde{G}_{i,22}^{-1}(\hat{G}_{i,22}-\tilde{G}_{i,22})\Vert_F)^{-1}\nonumber\\
        &\times(1 + \Vert B^T\tilde{P}_i\Vert_F) \nonumber\\
        &\leq 2\Vert R^{-1}\Vert_F(1 + \Vert B^T\tilde{P}_i\Vert_F).
\end{align}
where the second inequality comes from \cite[Inequality (5.8.2)]{horn2012matrix}, and the last inequality is due to \eqref{R_inverse_pertrubed_bound}. This completes the proof.
\end{proof}
\begin{lemma}\label{lemma_boundedness}
If
\begin{equation}\label{error_condition_bounded}
   \Vert \Delta G_i\Vert_F< (1+i^2)^{-1}a_i,\quad i=1,\cdots, \bar{i},
\end{equation}
where $a_i$ is defined in Lemma \ref{lemma_stability}, then
\begin{equation*}
    \Vert\tilde{P}_i\Vert_F\leq 6\Vert \tilde{P}_1\Vert_F,\quad \Vert\hat{K}_{i}\Vert_F\leq C_0,
\end{equation*}
for $i=1,\cdots,\bar{i}$, where
$$ C_0 = \max\left\{ \Vert\hat{K}_1 \Vert_F, 2\Vert R^{-1}\Vert_F\left(1+6\Vert B^T\Vert_F\Vert \tilde{P}_1\Vert_F\right)\right\}.$$
\end{lemma}
\begin{proof}
Inequality (\ref{key_inequality_2}) yields
\begin{equation}\label{key_inequality_3}
    \mathcal{L}_{\hat{A}_{i+1}}(\tilde{P}_i) + (Q+\hat{K}_{i+1}^TR\hat{K}_{i+1}) - \epsilon_{2,i}I < 0.
\end{equation}
where
$$\epsilon_{2,i} = \Vert \Delta G_i\Vert_F \mathbf{1}^T(\vert \mathcal{X}_{\hat{K}_{i}}\vert_{abs}+\vert \mathcal{X}_{\hat{K}_{i+1}}\vert_{abs})\mathbf{1}<1.$$
Inserting (\ref{eRPI_PE}) into above inequality, and using \ref{Lyapunov_monotone} in Lemma \ref{Lyapunov_properties}, we have
\begin{equation}\label{key_inequality_4_1}
    \tilde{P}_{i+1} < \tilde{P}_{i} + \epsilon_{2,i}\mathcal{L}_{\hat{A}_{i+1}}^{-1}(-I).
\end{equation}
With $Q>I_n$, (\ref{key_inequality_3}) yields
\begin{equation*}
    \mathcal{L}_{\hat{A}_{i+1}}(\tilde{P}_i) + (1 - \epsilon_{2,i})I < 0.
\end{equation*}
Similar to (\ref{key_inequality_4_1}), we have
\begin{equation}\label{key_inequality_4_2}
    \mathcal{L}_{\hat{A}_{i+1}}^{-1}(-I)<\frac{1}{1-\epsilon_{2,i}}\tilde{P}_i.
\end{equation}
From (\ref{key_inequality_4_1}) and (\ref{key_inequality_4_2}), we obtain
\begin{equation*}
    \tilde{P}_{i+1}<\left(1+\frac{\epsilon_{2,i}}{1-\epsilon_{2,i}}\right)\tilde{P}_i.
\end{equation*}
By definition of $\epsilon_{2,i}$ and condition \eqref{error_condition_bounded},
$$\frac{\epsilon_{2,i}}{1-\epsilon_{2,i}} \leq \frac{1}{i^2},\quad i=1,\cdots,\bar{i}.$$
Then \cite[Theorem 3]{InfinitSeries} yields
$$\tilde{P}_i\leq 6\tilde{P}_1,\quad i=1,\cdots,\bar{i}.$$
An application of \eqref{K_bound_by_P} completes the proof.
\end{proof}
Now we are ready to prove Lemma \ref{lemma_converge_to_neighbourhood}.
\begin{proof}[Proof of Lemma \ref{lemma_converge_to_neighbourhood}]
Consider Procedure \ref{procedure_robust_policy_iteration} confined to the first $\bar{i}$ iterations, where $\bar{i}$ is a sufficiently large integer to be determined later in this proof. Suppose
\begin{equation}\label{error_bound_2}
        \Vert\Delta G_i\Vert_F<b_{\bar{i}}\triangleq\frac{1}{2m(1+\bar{i}^2)}\left(\sqrt{n} + C_0\right)^{-2}.
\end{equation}
Condition (\ref{error_bound_2}) implies condition (\ref{error_condition_bounded}). Thus $\hat{K}_i$ is stabilizing, $\hat{G}_{i,22}^{-1}$ is invertible and $\tilde{P}_i\leq 6\Vert \tilde{P}_1\Vert_F$. By (\ref{eRPI_PE}) we have
\begin{equation*}
\begin{split}
\mathcal{L}_{\hat{A}_{i+1}}(\tilde{P}_{i+1}-\tilde{P}_{i})&= -Q-\hat{K}_{i+1}^TR\hat{K}_{i+1}-\mathcal{L}_{\hat{A}_{i+1}}(\tilde{P}_{i}).
\end{split}
\end{equation*}
Letting $E_i = \hat{K}_{i+1} - R^{-1}B^T\tilde{P}_i$, the above equation can be rewritten as
\begin{equation}\label{perturbed_newton}
    \tilde{P}_{i+1} = \tilde{P}_i - \mathcal{N}(\tilde{P_i}) + \mathcal{L}_{\mathcal{A}(\tilde{P_i})}^{-1}(\mathscr{E}_i),
\end{equation}
where $\mathcal{N}(\tilde{P_i}) = \mathcal{L}_{\mathcal{A}(\tilde{P_i})}^{-1}\circ\mathcal{R}(\tilde{P}_i),$ and 
\begin{equation*}
    \begin{split}
        \mathcal{R}(Y) &= Q + A^TY + YA - YBR^{-1}B^TY,  \\
    \mathscr{E}_i &= - E_i^TRE_i + \mathcal{L}_{BE_i}(\tilde{P}_{i+1}-\tilde{P}_{i}).
    \end{split}
\end{equation*}
Given $\hat{K}_1$, let $\mathcal{M}_{\bar{i}}$ denote the set of all possible $\tilde{P}_i$, generated by (\ref{perturbed_newton}) under condition (\ref{error_bound_2}). By definition, $\{\mathcal{M}_j\}_{j=1}^\infty$ is a nondecreasing sequence of sets, i.e., $\mathcal{M}_1\subset \mathcal{M}_2 \subset \cdots$. Define $\mathcal{M} = \cup_{j=1}^\infty\mathcal{M}_j$, $\mathcal{D} = \{P\in \mathbb{S}^n\ \vert\ \Vert P\Vert_F\leq 6\Vert \tilde{P}_1\Vert_F\}$. Then by Lemma \ref{lemma_boundedness} and Theorem \ref{theorem_standard_PI}, $\mathcal{M}\subset\mathcal{D}$; $\mathcal{M}$ is compact; $\mathcal{A}(P)$ is Hurwitz for any $P\in\mathcal{M}$.\par
Now we prove that for any $P^1,\ P^2\in \mathcal{M}$, there exists a constant $L$, such that
\begin{equation}\label{Newton_Lipschitz}
    \Vert\mathcal{N}(P^1)-\mathcal{N}(P^2)\Vert_F\leq L\Vert P^1-P^2\Vert_F.
\end{equation}
In \ref{Lyapunov_perturbation} of Lemma \ref{Lyapunov_properties}, let $X=\mathcal{A}(P^1)$, $X+\Delta X = \mathcal{A}(P^2)$, $Z = \mathcal{R}(P^1)$, $Z + \Delta Z = \mathcal{R}(P^2)$, $Y = -\mathcal{N}(P^1)$, $Y+\Delta Y = -\mathcal{N}(P^2)$, then, we have
\begin{align}\label{Newton_Lipschitz_inequality}
        &\Vert \Delta Y\Vert_F\leq\sqrt{n}\left\Vert \Delta Y\right\Vert_2 \nonumber\\
        &\leq \sqrt{n}\left(\Vert \mathcal{R}(P^1) - \mathcal{R}(P^2) \Vert_2\right. \nonumber\\
        &\left.+ 2\Vert \mathcal{A}(P^1) - \mathcal{A}(P^2)\Vert_2  \Vert \mathcal{N}(P^2) \Vert_2\right)\Vert H\Vert_2 \nonumber \\
        &\leq\sqrt{n} \left(C_1 + 2\Vert \mathcal{N}(P^2) \Vert_2\Vert BR^{-1}B^T\Vert_2\right)\Vert H\Vert_2\Vert P^1 -P^2 \Vert_2 \nonumber\\
        &\leq \sqrt{n}\left(C_1 + 2C_2\Vert BR^{-1}B^T\Vert_2\right)\Vert H\Vert_2\Vert P^1 -P^2 \Vert_2 \nonumber\\
        &\leq 6\sqrt{n}\Vert \tilde{P}_1\Vert_F\left(C_1 + 2C_2\Vert BR^{-1}B^T\Vert_2\right)\Vert P^1 -P^2 \Vert_2 \nonumber\\
        &\leq L \Vert P^1 -P^2 \Vert_F,
\end{align}
where $C_1$ is the Lipschitz constant of $\mathcal{R}(\cdot)$ on compact set $\mathcal{D}$. The fourth inequality in (\ref{Newton_Lipschitz_inequality}) is because of $\mathcal{M}\in\mathcal{D}$, thus $\Vert\mathcal{N}(P^2)\Vert_2\leq C_2$ for some $C_2>0$. For the fifth inequality in (\ref{Newton_Lipschitz_inequality}), by \ref{Lyapunov_solution} in Lemma \ref{Lyapunov_properties}
$$H = \mathcal{L}^{-1}_{\mathcal{A}(P^1)}(-I) = \int_0^\infty e^{\mathcal{A}^T(P^1)t}e^{\mathcal{A}(P^1)t}\dd t.$$
By definition of the set $\mathcal{M}$, there exist $G^1$ and stabilizing control gain $K^1$ associated with $P^1\in\mathcal{M}$, such that $\mathcal{H}(G^1,K^1)=0$. So for any $x\in\mathbb{R}^n$,
\begin{equation*}
    x^T\mathcal{R}(P_i)x = \min_{K\in\mathbb{R}^{m\times n}} x^T\mathcal{H}(G^1,K)x \leq 0.
\end{equation*}
This implies
$$\mathcal{L}_{\mathcal{A}(P^1)}(P^1)\leq -Q - P^1BR^{-1}B^TP^1<-I.$$
An application of \ref{Lyapunov_monotone} in Lemma \ref{Lyapunov_properties} to the above inequality leads to $H<P^1$. Thus we have
\begin{equation}\label{bound_on_H}
    \Vert H\Vert_F \leq 6\Vert \tilde{P}_1\Vert_F, \qquad \forall P^1\in\mathcal{M}\subset\mathcal{D},
\end{equation}
which implies the fifth inequality in (\ref{Newton_Lipschitz_inequality}).\par
Define $\{P_{k\vert i}\}_{k=0}^{\infty}$ as the sequence generated by \eqref{Kleinman_one_step} with $P_{0\vert i}=\tilde{P}_i$.
Similar to \eqref{perturbed_newton}, we have
\begin{equation}\label{exact_newton}
    P_{k+1\vert i} = P_{k\vert i} - \mathcal{N}(P_{k\vert i}), \quad k\in\mathbb{Z}_+.
\end{equation}
By Theorem \ref{theorem_standard_PI} and the fact that $\mathcal{M}$ is compact, there exists $k_0\in\mathbb{Z}_+$, such that 
\begin{equation}\label{triangle_inequality_first_part}
\Vert P_{k_0\vert i}-P^*\Vert_F<\delta_0/2, \qquad \forall P_{0\vert i}\in\mathcal{M}.
\end{equation}
Suppose
\begin{equation}\label{perturb_small}
    \Vert\mathcal{L}_{\mathcal{A}(\tilde{P}_{i+k})}^{-1}(\mathscr{E}_{i+k})\Vert_F<\mu,\qquad k=0,\cdots, \bar{i} - i.
\end{equation}
We find an upper bound on $\Vert P_{k\vert i}-\tilde{P}_{i+k}\Vert_F$. Notice that from (\ref{perturbed_newton}) and (\ref{exact_newton}),
\begin{equation*}
    \begin{split}
        P_{k\vert i} &= P_{0\vert i} - \sum_{j=0}^{k-1} \mathcal{N}(P_{j\vert i})  \\
        \tilde{P}_{i+k} &= \tilde{P}_i - \sum_{j=0}^{k-1} \mathcal{N}(\tilde{P}_{i+j}) + \sum_{j=0}^{k-1} \mathcal{L}_{\mathcal{A}(\tilde{P}_{i+j})}^{-1}(\mathscr{E}_{i+j}).
    \end{split}
\end{equation*}
Then (\ref{Newton_Lipschitz}) and (\ref{perturb_small}) yield 
\begin{equation*}
    \begin{split}
        \Vert P_{k\vert i} - \tilde{P}_{i+k}\Vert_F \leq k\mu + \sum_{j=0}^{k-1} L \Vert P_{j\vert i} - \tilde{P}_{i+j}\Vert_F.
    \end{split}
\end{equation*}
An application of the Gronwall inequality \cite[Theorem 4.1.1.]{agarwal2000difference} to the above inequality implies
\begin{equation}\label{continuous_dependence}
    \Vert P_{k\vert i} - \tilde{P}_{i+k}\Vert_F \leq k\mu + L\mu \sum_{j=0}^{k-1}j (1+L)^{k-j-1}.
\end{equation}
The error term in (\ref{perturbed_newton}) satisfies
\begin{align}
    \left\Vert \mathcal{L}_{\mathcal{A}(\tilde{P_i})}^{-1}(\mathscr{E}_i)\right\Vert_F 
    &= \left\Vert \int_0^\infty e^{\mathcal{A}^T(\tilde{P}_i)t}\otimes e^{\mathcal{A}(\tilde{P}_i)t} \dd t \vect\left( \mathscr{E}_i\right)\right\Vert_2 \nonumber \\
    &\leq C_3\Vert \mathscr{E}_i\Vert_F,\label{perturbed_inverse_lypunov}
\end{align}
where $C_3$ is a constant and the inequality is due to (\ref{bound_on_H}). Let $\bar{i}>k_0$, and $k=k_0$, $i = \bar{i}-k_0$ in \eqref{continuous_dependence}. Then by condition \eqref{error_bound_2}, Lemma \ref{lemma_boundedness}, (\ref{perturb_small}), and (\ref{continuous_dependence}), there exists $i_0\in\mathbb{Z}_+$, $i_0>k_0$, such that $\Vert P_{k_0\vert \bar{i}-k_0} -\tilde{P}_{\bar{i}}\Vert_F<\delta_0/2$, for all $\bar{i}\geq i_0$. Setting $i = \bar{i}-k_0$ in \eqref{triangle_inequality_first_part}, the triangle inequality yields $\tilde{P}_{\bar{i}}\in\mathcal{B}_{\delta_0}(P^*)$, for $\bar{i}\geq i_0$. Then in \eqref{error_bound_2}, choosing $\bar{i}\geq i_0$ such that $\delta_2 = b_{\bar{i}}<\min(\gamma^{-1}(\epsilon),\delta_1)$ completes the proof.
\end{proof}

\subsection{Proof of Theorem \ref{theorem_data_driven_robut_PI}}
Let $x = \tilde{x}$, $u = \hat{u}$ in \eqref{ADP_LTI}, where $\tilde{x}$ denote the state trajectory of system \eqref{LTI_sys} when the input is $\hat{u}$ and $x(0) = \hat{x}(0)$. Since $Q$, $R$ are known, we have
\begin{equation}\label{ADP_LTI_tilde}
   \tilde{\mathcal{Y}}_i \triangleq \left[\begin{array}{cc}
         \svec(\tilde{P}_i)  \\
         \vect(\tilde{K}_{i+1}) 
    \end{array}\right] = \Theta(\hat{K_i})^\dagger\Xi(\hat{K}_i),
\end{equation}
where $\tilde{K}_{i+1} = R^{-1}B^T\tilde{P}_i$ and
$$\hat{G}_i = \left[\begin{array}{cc}
        Q + A^T\hat{P}_i + \hat{P}_iA & \hat{K}^T_{i+1}R \\
        R\hat{K}_{i+1} & R
        \end{array}\right].$$
Define $\hat{\mathcal{Y}}_i = [\svec^T(\hat{P}_i),\vect^T(\hat{K}_{i+1})]^T$. It is easy to prove $\Vert \Delta G_i\Vert_F\leq \gamma_1(\Vert\hat{\mathcal{Y}}_i - \tilde{\mathcal{Y}}_i\Vert_2)$, where $\gamma_1\in\mathcal{K}$. Thus the proof amounts to showing that there exists a $c$, such that 
\begin{equation}\label{equivalent_proof_of_ADP_continuous}
    \Vert \hat{\mathcal{Y}}_i - \tilde{\mathcal{Y}}_i\Vert_2<\gamma_1^{-1}(\delta_2),\  \forall \Vert w\Vert_\infty<c, \  \forall i\in\mathbb{Z}_+, \ i>0.
\end{equation}
By Corollary \ref{corollary_global}, $\Vert \tilde{P}_i\Vert_F<M_0$, $\Vert\hat{K}_{i}\Vert_F<M_1$ for some $M_1>0$, if $\Vert \Delta G_i\Vert_F<\delta_2$. Let $M=\max(M_0,M_1)$, and define set
\begin{equation*}
    \begin{split}
        \mathcal{F} = \{K\in\mathbb{R}^{m\times n}&\vert \Vert K\Vert_F\leq M, \Vert P_K\Vert_F\leq M, \\
        &(A-BK) \text{ is Hurwitz}\},
    \end{split}
\end{equation*}
where $P_K = \mathcal{L}^{-1}_{A-BK}(-Q-K^TRK)$. Then to fulfill \eqref{equivalent_proof_of_ADP_continuous}, it is sufficient to prove
\begin{equation}\label{sensitivity_perturbation_least_square}
      \Vert \Theta(K)^\dagger\Xi(K) - \hat{\Theta}(K)^\dagger\hat{\Xi}(K) \Vert_2
      < \gamma_1^{-1}(\delta_2),
\end{equation}
for all $K\in\mathcal{F}$, and $\Vert w\Vert_\infty<c$.

By \cite[Lemma 6]{JIANG20122699}, $\Theta(K)$ has full column rank for all $K\in\mathcal{F}$. By linear system theory \cite{wonham1974linear},
\begin{equation}\label{data_error_can_be_small}
    \Vert [I_{xx},I_{xu}] - [I_{\hat{x}\hat{x}},I_{\hat{x}\hat{u}}]\Vert_F<\gamma_2(\Vert w\Vert_\infty),\quad \gamma_2\in\mathcal{K}.
\end{equation}
Then by continuity, and the fact that $\mathcal{F}$ is compact (otherwise $\Vert P_K\Vert_F\leq M$ cannot be true), there exists $c'>0$ such that for any $\Vert w\Vert_\infty<c'$ and any $K\in\mathcal{F}$, $\hat{\Theta}(K)$ has full column rank, and
\begin{equation*}
    \begin{split}
        \Vert \Theta(K)\Vert_2&>\alpha_1, \quad \Vert \Theta(K)^\dagger\Xi(K)\Vert_2<\alpha_2,\quad \\ \kappa(\Theta(K))&<\alpha_3, \quad \kappa(\Theta(K)) \Vert \Theta(K) - \hat{\Theta}(K)\Vert_2<0.5\alpha_1,
    \end{split}
\end{equation*}
where $\alpha_1$, $\alpha_2$ and $\alpha_3$ are all positive finite constants, and $\kappa(\Theta(K))$ is condition number of $\Theta(K)$. Then by \cite[Theorem 1.4.6.]{doi:10.1137/1.9781611971484}, we have
\begin{equation*}
    \begin{split}
      &\Vert \Theta(K)^\dagger\Xi(K) - \hat{\Theta}(K)^\dagger\hat{\Xi}(K) \Vert_2
      < 2\alpha_3\alpha_1^{-1}\times\\
      &\left(\alpha_2\Vert \Theta(K) - \hat{\Theta}(K) \Vert_F + \Vert \Xi(K) - \hat{\Xi}(K)\Vert_F\right),
    \end{split}
\end{equation*}
for any $K\in\mathcal{F}$. By \eqref{data_error_can_be_small}, we can always choose a $c<c'$, such that \eqref{sensitivity_perturbation_least_square} is satisfied, which completes the proof.
\bibliography{RPIRef}
\bibliographystyle{IEEEtran}

\end{document}